\documentclass[letterpaper, conference]{ieeeconf}
\IEEEoverridecommandlockouts
\RequirePackage[font=footnotesize]{caption}
\usepackage{subcaption}
\usepackage{xcolor}

\usepackage{amsmath,amssymb}
\usepackage{upgreek}
\newlength{\dhatheight}

\usepackage{mathtools}
\usepackage{bm}
\usepackage{yhmath}
\usepackage{mathrsfs}
\newcommand{\R}{\mathbb R}
\newcommand{\X}{\mathcal X}

\usepackage{cuted}
\newtheorem{theorem}{Theorem}[section]
\newtheorem{definition}[theorem]{Definition}

\newtheorem{lemma}[theorem]{Lemma}

\newtheorem{remark}[theorem]{Remark}

\newtheorem{proposition}[theorem]{Proposition}
\newtheorem{problem}{Problem}
\newtheorem{conjecture}{Conjecture}

\newcommand{\longthmtitle}[1]{\mbox{}{\textit{(#1):}}}

\usepackage{apptools}
\AtAppendix{\counterwithin{theorem}{section}}

\usepackage{algorithmic}
\usepackage{enumerate}

\usepackage{array}

\usepackage{url}

\definecolor{gnblue6}{RGB}{35,156,255}
\newcommand{\red}[1]{\textcolor{red}{#1}}

\newcommand{\real}{\mathbb{R}}

\renewcommand{\real}{\mathbb{R}}

\newcommand{\mc}{\mathcal}

\newcommand{\diag}{\mathrm{diag}}

\newcommand{\argmin}[2] {\mathrm{arg}\min_{#1}#2}

\DeclareSymbolFont{bbold}{U}{bbold}{m}{n}
\DeclareSymbolFontAlphabet{\mathbbold}{bbold}

\begin{document}

\title{Timescale Limits of Linear-Threshold Networks}

\author{
William Retnaraj \quad Simone Betteti \quad Alexander Davydov \quad Francesco Bullo \quad Jorge Cort\'es\thanks{\red{Add thanks.} W. Retnaraj and J. Cort\'es are with the University of California San Diego, {\tt \footnotesize \{wretnaraj,cortes\}@ucsd.edu}, S. Betteti is with the Italian Institute of Artificial Intelligence for Industry, {\tt \footnotesize simone.betteti@ai4i.it}, A. Davydov is with Rice University, {\tt \footnotesize davydov@rice.edu}, and F. Bullo is with the University of California Santa Barbara,
{\tt \footnotesize bullo@ucsb.edu}.
}}
\maketitle

\begin{abstract}
Linear-threshold networks (LTNs) capture the mesoscale behavior of interacting populations of neurons and are of particular interest to control theorists due to their dynamical richness and relative ease of analysis. 
The aim of this paper is to advance the study of global asymptotic stability in LTNs with asymmetric neural interactions and heterogeneous dissipation under the structural Lyapunov diagonal stability (LDS) condition. 
To this end, we introduce a one-parameter family of LTNs that preserves the LDS condition and has a parameter-independent equilibrium set. 
In the fast limit, this family converges to a projected dynamical system (PDS), while in the slow limit, it converges to a discontinuous hard-selector system (HSS).
Under LDS, we prove that the fast PDS limit is globally exponentially stable and that the HSS limit is globally asymptotically stable. 
This alignment suggests that the limiting systems capture essential mechanisms governing stability across the entire LTN family. 
Together with numerical evidence, these findings indicate that resolving stability at the fast and slow endpoints provides a promising and structurally grounded path toward establishing global stability for LTNs with biologically plausible recurrence and diagonal dissipation.
\end{abstract}

\section{Introduction}
Autonomous dynamical systems provide a foundational framework for modeling neural computation, capturing recurrent interactions through a combination of a \emph{dissipative component} that contracts activity and a \emph{synaptic component} that encodes nonlinear recurrent coupling.
Classical results for systems with \emph{symmetric} synaptic matrices, originating from Grossberg~\cite{CMA-GS:83} and Hopfield~\cite{HJJ:82,HJJ:84}, exploit symmetry to construct Lyapunov energy functions, ensuring dissipation and convergence to an equilibrium. 
However, symmetry imposes constraints often incompatible with biologically-realistic connectivity.

The \emph{linear-threshold network} (LTN) is a tractable model that allows for asymmetric synaptic capturing, resulting in a wide variety of dynamical behavior, including stable equilibria~\cite{EN-JC:21}, multistability, oscillations~\cite{EN-RP-JC:22}, and chaos~\cite{KM-AD-VI-CC:24}.
Stability results for the LTN remain limited and often rely on restrictive assumptions such as total \(\mc L\)-stability~\cite{EN-JC:21}. 
A more general structural condition is \emph{Lyapunov diagonal stability} (LDS)~\cite{EK-AB:00}, which has been studied in the context of various dissipative network systems~\cite{MA-CM-AP:16}.
For the closely-related Hopfield network, LDS guarantees global stability via a Lyapunov function~\cite{FM-TA:95}.
Since the appearance of~\cite{FM-TA:95} in 1995, it has been a long-standing conjecture that LDS is sufficient for the stability of firing-rate networks, of which LTN is a special case.
Existing proofs that use LDS, however, rely on scalar dissipation and do not extend to the biologically relevant case of arbitrary \emph{diagonal} dissipation. 

Recent numerical evidence~\cite{BS-RW-DA:26} suggests that LDS may suffice to ensure global stability for firing-rate networks.
In particular, we focus on two relevant and limiting cases of linear-threshold dynamics: the projected dynamical system (PDS; fast limit) and the hard-selector system (HSS; slow limit). 
Conceptually, PDS reveals the stabilizing role of componentwise dissipative dominance, while HSS reveals the stabilizing role of nonlinear switching-regime boundary orientations.
In this work, we show that both systems inherit strong dissipativity and stability properties under LDS. 
In particular, we introduce suitable Lyapunov energy functions for both fast and slow limits, and prove global stability of the respective dynamics. 
The fact that both limiting mechanisms are stable under the same diagonal condition strongly suggests that the intermediate LTNs might also be stable. 
Thus, the analysis of the limiting systems is not only analytically convenient, but it gives us conceptually fundamental building blocks toward a general stability proof.

\textbf{Contributions:} In this paper, we contribute to the analysis of global stability in asymmetric neural networks as follows:
\begin{enumerate}[(i)]
    \item we introduce a parametrized family of linear-threshold networks, named the $\tau$-LTN family, preserving Lyapunov diagonal stability for all $\tau$;
    \item we show how the limiting endpoint systems, fast and slow, of the family correspond to projected dynamics and hard-selector dynamics, respectively, and establish uniqueness of equilibria for all systems;
    \item we prove global stability of the unique equilibrium of the fast-limit and slow-limit systems, respectively, by introducing suitable Lyapunov functions.
\end{enumerate}
Section~II introduces notation and preliminaries. Section~III presents the \(\tau\)-LTN family and formulates the stability problem. Section~IV establishes equilibrium properties and preservation of diagonal stability. Sections~V and~VI analyze the fast and slow-limit systems. Section~VII is a discussion with some numerical exploration, and Section~VIII concludes with a summary of how the \(\tau\)-LTN family and the endpoints study helps us make progress towards the global stability of LTNs under Lyapunov diagonal stability, and we posit future research directions.

\section{Preliminaries}
\label{sec:SIprelim}
\subsection{Notation}
Define the set \([n]=\{1,\dots, n\}\) for \(n>0\).
Denote the set of all reals as \(\R\), the set of all positive reals as \(\R_+\), and the \(n\)-dimensional Euclidean space by \(\R^n\).
For this section, let \(z\in \R\) be a scalar and let \(x, y \in \mathbb R^n\) be vectors.
Denote \(x = [x_1,\dots, x_n]^{\top}\).
Let \([z]_0^m = \max{(0, \min{(z, m)})}\) with \(m>0\).
When applied to \(x\in \R^n\), \([x]_0^m\) is to be understood elementwise.
Define the positive (resp. negative) part of \(z\) as \((z)_+ = \max(0,\, z)\) (resp. \((z)_- = -\min(0,\, z)\)). 
Define \(\bf 1_n\) to be the \(n\)-dimensional all-ones vector.
Let \(\X\subset \R^n\).
Define \({\bf 1_\X}:\R^n \rightarrow \{0, 1\}\) be the \(0\)--\(1\) indicator of \(\X\), so that \({\bf 1_\X}(x) = 1\) if \(x \in \X\), and \(0\) if not.

Let \(\mathbb D^n_+\) be the space of all positive diagonal matrices.
For this section, let \(A \in \R^{n\times n}\) be a real-valued square matrix.
Let \(A_i \in \R^n\) represent the \(i\)th row of \(A\) so that \(A = [A_1, \dots A_n]^\top\).
Let \(A_{ij} \in \R\) represent the entry in the \(i\)th row and \(j\)th column of \(A\).
Denote \(\diag(d_1, \dots, d_n)\) as the diagonal matrix with scalars \(d_1, \dots, d_n\) on the diagonal.
We denote the identity matrix in \(\mathbb R^{n\times n}\) by \(I_n\), and drop the subscript when the dimension is clear.
By \(A\succ 0\) (resp. \(A \prec 0\)), we mean that \(A\) is positive definite (resp. negative definite).
Let \(\lVert x \rVert_P\) be the \(P\)-weighted Euclidean 2-norm of \(x\) for \(P \succ 0\).
If \(P= I\), the above becomes \(\lVert x \rVert_2\) or simply \(\lVert x \rVert\), the Euclidean 2-norm.
We say that \(A\) is Lyapunov diagonally stable (\(A \in \mc LDS\)) if \(A^\top \Lambda + \Lambda A \prec 0\) holds for some \(\Lambda \in \mathbb D^n_+\), which is called the diagonal certificate.

Let \(\mathfrak B(X)\) be the collection of all subsets of set \(X\).
Define the convexified hard-selector map \(h:\R \rightarrow \mathfrak B([0,1])\) and its vector equivalent \(\mc H: \R^n \rightarrow \mathfrak B([0,1]^n)\) as
\begin{equation}\label{eq:hard_sel_cvx}
    h(z) = 
    \begin{cases}
    \{0\}, & z<0,\\
    [0,1], & z=0,\\
    \{1\}, & z>0;
    \end{cases}
    \qquad
    \mc H(x) 
    = 
    \begin{bmatrix}
        h(x_1)\\ \vdots\\ h(x_n)
    \end{bmatrix}.
\end{equation}

\subsection{Projected dynamical systems}
Here, we describe basic notions on projected dynamical systems following~\cite{PD-AN:93,AN-DZ:96}.
Let \(\X \subset \mathbb R^n\) be a closed convex polyhedron and let \(x, v \in \mathbb R^n\) be two vectors.
If \(x \in \mathcal X\), we define the tangent cone of \(\X\) at \(x\ \in \X\) as \(T_\X (x)\).
Intutively, the tangent cone is the set of all directions that do not leave \(\X\), if one steps in their direction infinitesimally from \(x \in \X\).
Define the projection operator \(\Pi_\mathcal{X}: \mathbb R^n \times \mathbb R^n \rightrightarrows \mathbb R^n\) by
\begin{equation}\label{eq:projection_op}
    \Pi_\mathcal{X}(x, v) = \argmin{y\in T_\mathcal{X}(x)}{\lVert y - v\rVert^2}.
\end{equation}
Notice that the projection onto a closed convex cone (arising from the convexity of \(\X\)) is single-valued. 
For any vector field \(X:\mathbb R^n \rightarrow \mathbb R^n\), with slight abuse of notation, construct the following dynamical system that is a nonclassical ODE
\begin{equation}\label{eq:PDS_defn}
    \dot x = \Pi_\mathcal{X}(x, X(x)).
\end{equation}
We call~\eqref{eq:PDS_defn} a projected dynamical system (PDS).
Absolutely continuous forward solutions that satisfy~\eqref{eq:PDS_defn} a.e. in time exist and are unique for Lipschitz vector field \(X\) and convex polyhedron \(\X\)~\cite[Theorem 2]{PD-AN:93}.

\subsection{Discontinuous dynamical systems}
Here, we provide basic background on  discontinuous systems following the exposition of~\cite{JC:08} and results from~\cite{AFF:88}.
Consider the dynamics
\begin{equation}\label{eq:discontinuous_ode}
    \dot x = X(x)
\end{equation}
where \(X:\X \rightarrow \R^n\) is a (possibly discontinuous) piecewise continuous vector field and \(\X\subseteq \mathbb R^n\) is a convex polyhedron.
A Filippov solution of~\eqref{eq:discontinuous_ode} is an absolutely continuous function \(x:[0, \infty)\rightarrow \R^n\) that satisfies
\begin{equation}\label{eq:F[discontinuous_ode]}
    \dot x(t) \in \operatorname{F}[X] (x(t)), \qquad\text{for a.e.}\ t \geq 0
\end{equation}
where \(\operatorname{F}[X]:\R^n \rightarrow \mathfrak{B}(\R^n)\) is the Filippov regularization of \(X\) given by
\begin{equation}\label{eq:filippov_reg}
    \operatorname{F}[X](x) \coloneqq \bigcap_{\delta>0}\bigcap_{\mu(S)=0} \overline{\operatorname{co}}\{X(B(x, \delta)\setminus S)\}
\end{equation}
for all \(x \in \X\), where \(\overline{\operatorname{co}}\) is the convex closure of its argument, \(\mu\) is the Lebesgue measure, and \(B(x, \delta)\) represents the open ball centered at \(x\in \R^n\) with radius \(\delta>0\).
The intuition behind Filippov solutions is therefore to examine the immediate neighborhood of each point \(x\) and get a set of admissible trajectory directions.

Filippov solutions of~\eqref{eq:discontinuous_ode} exist if \(X\) is measurable and locally essentially bounded~\cite{AFF:88},~\cite[Proposition 3]{JC:08}.
Uniqueness of Filippov solutions,
however, generally fails for piecewise continuous vector fields.
That does not prevent us from studying properties of trajectories, such as the stability of an equilibrium.
We attach the adjective ``strong" to any property when we mean that all solutions starting from each initial condition satisfy it, whereas the adjective ``weak" means a solution from each initial condition satisfies it.

\section{Problem setup}
Linear-threshold networks (LTNs) model the autonomous dynamics of recurrent neural populations through a biologically motivated yet tractable structure. Let \(D \in \mathbb{D}^n_+\) denote the diagonal dissipation matrix and \(W \in \mathbb{R}^{n \times n}\) the synaptic matrix. The LTN dynamics are
\begin{equation}\label{eq:LTN}\tag{LTN}
    \dot x = -Dx + [Wx + u]_0^1,
\end{equation}
where \(x\) represents average population firing rates. Under this model, the set \(\mathcal{X} = \{x \in \mathbb{R}^n : Dx \in [0,1]^n\}\) is forward invariant (cf. Lemma~\ref{lem:forward_invariance}), ensuring nonnegative and bounded activity. The matrix \(W\) encodes interaction strengths between populations, while \(D\) governs leak dynamics; allowing \(D \neq I\) captures heterogeneous dissipation across populations. The interplay between \(W\) and \(D\) determines the overall stability of the network. Inspired by the study in~\cite{FM-TA:95}, we investigate how structural properties of \(W - D\) relate to global stability of~\eqref{eq:LTN}, focusing on \emph{Lyapunov diagonal stability} (LDS) as a candidate sufficient condition for global asymptotic stability
(GAS). This motivates the following conjecture.

\begin{conjecture}\longthmtitle{LDS implies GAS of LTN equilibria}
\label{conj:LDS-GAS}
    Consider~\eqref{eq:LTN} with \(D \in \mathbb{D}^n_+\), \(W \in \mathbb{R}^{n\times n}\).
    If \(W - D \in \mathcal{LDS}\), then for every \(u \in \mathbb{R}^n\),
   ~\eqref{eq:LTN} admits a unique globally asymptotically stable equilibrium.
\end{conjecture}
Conjecture~\ref{conj:LDS-GAS} is strongly supported by Monte-Carlo simulations in~\cite{BS-RW-DA:26}, which report a \(100\%\) GAS success rate for dimensions \(n = 3,4,5\) and for activation functions extending beyond the linear-threshold case. However, analytical proofs for \(D \neq I\) remain elusive: classical Lur’e-Postnikov Lyapunov constructions~\cite{REK:63}
do not accommodate heterogeneous dissipation well.

To uncover the mechanisms that may enforce global stability, we isolate two structural features intrinsic to LTNs: (i) the dissipativity encoded by the LDS property of \(W - D\), and (ii) the nonlinear switching induced by the saturation of \([\,\cdot\,]_0^1\). These components become analytically transparent in two limiting regimes of LTN behavior: a \emph{fast} regime dominated by projection onto the invariant set, and a \emph{slow} regime  dominated by saturated (hard-selecting) dynamics. Studying these limits enables us to untangle and analyze the two motifs independently. To formalize this approach, we construct a one-parameter family of dynamics whose extremes correspond to these limiting behaviors, while the original LTN is recovered at a specific intermediate value. This framework allows the fast and slow limits to reveal the underlying stability mechanisms that we aim to extend to general LTNs.
\begin{definition}[$\tau$-LTN family]\label{def:tLTN}
    Let \(\tau>0\) and define
    \begin{equation}\label{eq:t-LTN_field}
        x\mapsto f_\tau(x):=\frac1\tau\left(-Dx+[Dx+\tau(Ax+u)]_0^1\right),
    \end{equation}
    as the $\tau$-LTN map for \(x\ \in \X\), where \(A:=W-D\). 
    The $\tau$-LTN family of dynamics is given by
    \begin{equation}\label{eq:t-LTN}\tag{$\tau$-LTN}
         \dot x=f_\tau(x).
    \end{equation}
\end{definition}

Notice that \emph{each} system in this family is itself an LTN with synaptic matrix $W_\tau=(1-\tau)D+\tau W$, the same diagonal dissipation, and therefore the same state space $\X$, but evolving on a timescale $\tau>0$ and driven by a scaled input $\tau u$. 
As shown later in Section~\ref{sec:t-LTN}, all members of~\eqref{eq:t-LTN} share the same set of equilibria.

\begin{figure}
    \centering
    \includegraphics[width=1\linewidth]{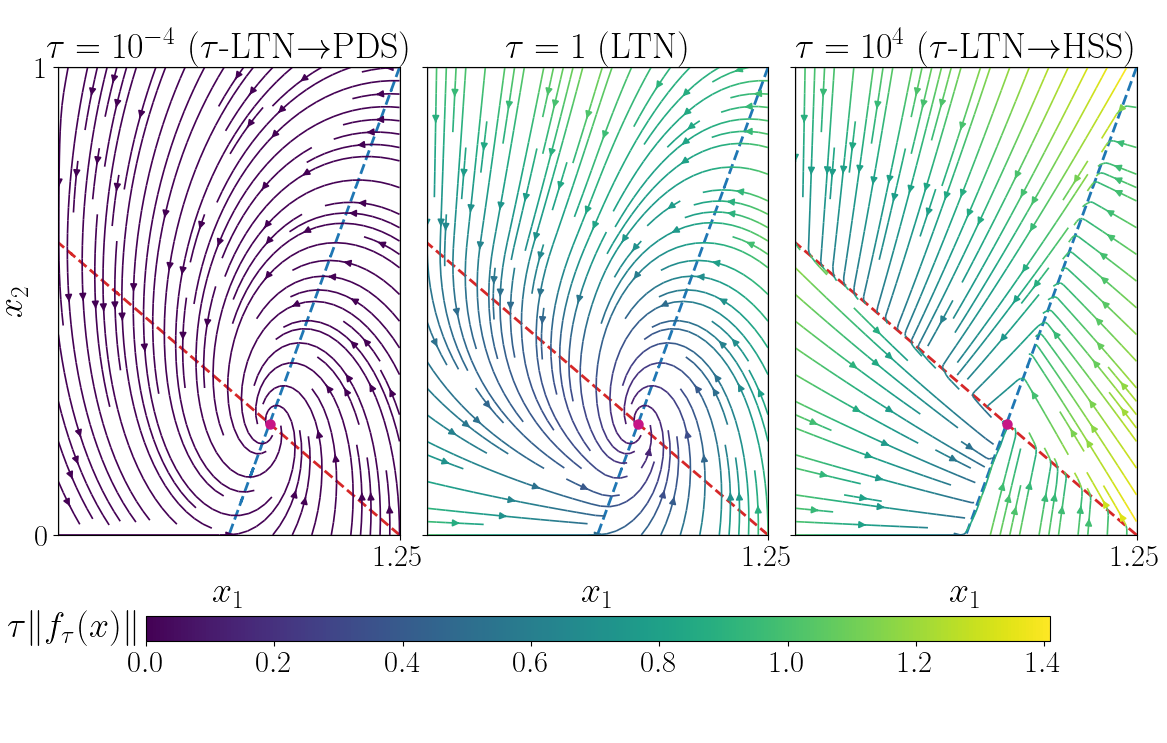}
    \caption{Vector field streamplots of the LDS-preserving family ($\tau$-LTN) of linear-threshold networks (LTNs) with \(W=[0.0,\, -1.6;\, 1.6,\, 0.0]\),
    \(D = \diag(0.8,\,1.0)\), and \(u = [1,\, -1]^\top\),
    and thus \(W-D \in \mc LDS\) with diagonal certificate \(I_2\),
    for \(\tau\in \{10^{-4},\,1,\,10^{4}\}\). The left panel approaches fast projected dynamical system (PDS) like behavior (flow almost `slams' into the boundary); the right panel approaches the slow-limit hard-selector system (HSS) like behavior (almost discontinuous change in flow direction); the middle panel is the canonical LTN. 
    Empirically, \(\tau=1\) blends qualitative vector-field features of both extremes. The violet dot denotes the common equilibrium; the red dashed line depicts the locus of \((Ax + u)_1 = 0\); the blue one of \((Ax + u)_2 = 0\). Color bar: 2-norm of the \(\tau\)-scaled vector field.}
    \label{fig:t-LTN_panels1}
\end{figure}

The \(\tau\)-LTN construction admits two limiting regimes: a fast-timescale limit at $\tau\to 0^+$ and a slow-timescale limit at $\tau\to+\infty$. The fast limit yields the projected dynamical system (PDS) $\dot{x}=\Pi_{\X}(x, Ax+u)$, while the slow limit belongs to the hard-selector system (HSS) $\dot{x}\in-Dx+\mc H(Ax+u)$, with \(\mc H\) switching according to the sign of its argument, as defined in~\eqref{eq:hard_sel_cvx}. 
Neither limiting system belongs to the family, yet both share the same equilibrium set as every $\tau$-LTN. 

Figure~\ref{fig:t-LTN_panels1} illustrates the vector fields for three representative values of $\tau$ with fixed $D\neq I$ and $A\in\mc LDS$, comparing dynamics near both asymptotic limits with the canonical LTN. 
Overall, the family~\eqref{eq:t-LTN} provides a unified framework that preserves equilibria across timescales and facilitates the stability analysis of LTNs. We are now able to pose our problem of interest, which tractably encodes Conjecture~\ref{conj:LDS-GAS}.

\begin{problem}\longthmtitle{LTN global stability via endpoint analysis}
    Consider the family of dynamical systems with \(\tau > 0\)
    \[\dot x = \frac1\tau\left(-Dx+[Dx+\tau(Ax+u)]_0^1\right)\]
    with \(u \in \R^n\), dissipation matrix \(D \in \mathbb D_+^n\) and synaptic matrix \(W \in \R^{n\times n}\).
    Also consider the limiting dynamical systems that arise when \(\tau \to 0^+\) (fast endpoint) or \(\tau \to +\infty\) (slow endpoint).
    Let \(A = W-D\).
    Then, show that if \(A \in \mc LDS\), both the fast and slow limits have the same equilibrium point that is globally stable. 
    Using the endpoint stability, under the same hypothesis, prove that each member of the family has a GAS equilibrium.
    In particular, show that the \(\tau = 1\) member, i.e.,~\eqref{eq:LTN} has a GAS equilibrium.
\end{problem}

Importantly, we remark that in this work, we do not solve the full problem as stated above, but make promising progress towards proving Conjecture~\ref{conj:LDS-GAS}. While we solve the endpoint stability problem and isolate the respective Lyapunov stability mechanisms, we leave the bridge to the global stability of the \(\tau\)-LTN family to future work.

\section{The $\tau$-LTN family: equilibrium and structural stability}\label{sec:t-LTN}
Consider the parametric family of LTNs given in Definition~\ref{def:tLTN}. 
As we will point out later, the family is highly regular; but most notably for our use, and what we establish below, is that if one member is LDS, each member is LDS.
\begin{lemma}[Preservation of LDS]
    \label{lem:LDS_preserved}
    For every \(\tau>0\), the member~\eqref{eq:t-LTN} is LDS if and only if \(A\in\mc LDS\).
\end{lemma}
\begin{proof}
The effective synaptic matrix of~\eqref{eq:t-LTN} is \(W_\tau=(1-\tau)D+\tau W\),
so \(W_\tau-D=\tau(W-D)=\tau A\), while the dissipation matrix remains the same.
Thus,
\[(W_\tau-D)^\top\Lambda+\Lambda(W_\tau-D) =
\tau(A^\top\Lambda+\Lambda A),\]
and thus for each \(\tau >0\) the diagonal Lyapunov inequality holds for \(W_\tau-D\) exactly when it holds for \(A\).
\end{proof}

Under $A=W-D\in\mc LDS$, we next prove that the \(\tau=1\) member, i.e.,~\eqref{eq:LTN}, admits a unique equilibrium point in~$\X$. 
\begin{proposition}\longthmtitle{Unique LTN equilibrium under LDS}\label{prop:LTN-eq}
    Define $x\mapsto f_1(x)=-Dx+[Wx+u]_0^1$, for fixed external input $u\in\R^n$, and suppose that $W-D\in\mc LDS$. 
    Then there exists a unique $x^\star\in\real^n$ such that $f_1(x^\star)=0$.
    In particular, for each \(u\), under the LDS assumption,~\eqref{eq:LTN} admits a unique equilibrium.
\end{proposition}
\begin{proof}
Under the LDS assumption, by~\cite[Lemma 2]{FM-TA:95}, for every diagonal \(K\) such that \(0\preceq K\preceq I_n\), the matrix \(D-KW\) is nonsingular.
Define \(T(x):=D^{-1}[Wx+u]_0^1\). Then \(f_1(x)=0\) if and only if \(x=T(x)\). 
Since \([Wx+u]_0^1\in[0,1]^n\) for all \(x\in\real^n\), the map \(T\) sends the compact convex set \(\X\) into itself. 
Hence, by Brouwer's fixed-point theorem~\cite{brouwer1911abbildung}, \(T\) has a fixed point \(x^\star\in \X\), and therefore \(f_1(x^\star)=0\).
For uniqueness, let \(x,y\in\real^n\) satisfy \(f_1(x)=f_1(y)=0\). 
Then \(Dx=[Wx+u]_0^1\) and \(Dy=[Wy+u]_0^1\). For each \(i\in\{1,\dots,n\}\), set \(a_i:=(Wx+u)_i\), \(b_i:=(Wy+u)_i\), and define \(k_i:=([a_i]_0^1-[b_i]_0^1)/(a_i-b_i)\) when \(a_i\neq b_i\), and \(k_i:=0\) otherwise. 
Since \(z\mapsto [z]_0^1\) is nondecreasing and \(1\)-Lipschitz, \(k_i\in[0,1]\) for all \(i\). Let \(K:=\diag(k_1,\dots,k_n)\). 
Then \([Wx+u]_0^1-[Wy+u]_0^1=KW(x-y)\), so \((D-KW)(x-y)=0\). 
Recall that \(D-KW\) is nonsingular; hence \(x=y\).
Therefore, there exists a unique \(x^\star\in\real^n\) such that \(f_1(x^\star)=0\).
\end{proof}

We show that all members of the family have the same equilibria, and extend the above result to the whole family.
\begin{proposition}\longthmtitle{Common equilibrium set}
\label{prop:common_equilibria}
Consider the family~\eqref{eq:t-LTN} with \(D \in \mathbb D_+^n\), \(W \in \R^{n\times n}\).
Then,
\begin{enumerate}[(i)]
    \item the equilibria of the family are independent of \(\tau>0\).
    \item if \(A = W - D \in \mc LDS\), the family possesses a unique equilibrium.
\end{enumerate}
\end{proposition}

\begin{proof}
Let \(x \in \X\). 
(i) The equilibrium condition \(f_\tau(x)=0\) is \(Dx=[Dx+\tau(Ax+u)]_0^1\).
Using the definition of the saturation function for each case,  we get for \(i \in [n]\),
\begin{align}\label{eq:eq_cond_LTN}
    \begin{cases}
    (Ax+u)_i \le 0, & d_ix_i=0,\\
    (Ax+u)_i = 0, & 0<d_i x_i<1,\\
    (Ax+u)_i \ge 0, & d_i x_i=1.
    \end{cases}    
\end{align}
Observe that~\eqref{eq:eq_cond_LTN} does not depend on \(\tau\), which proves (i).
(ii) Assume \(A \in \mc LDS\). 
By Proposition~\ref{prop:LTN-eq}, the \(\tau=1\) member admits a unique equilibrium in \(\X\).
By (i), each member therefore also has a unique equilibrium in \(\X\).
\end{proof}

\section{Fast limit: projected dynamical system}
Here we study the global stability properties of the fast limit of the $\tau$-LTN family.
Define the projected dynamical system (PDS) with \(x \in \X\) as
\begin{equation}\label{eq:PDS}\tag{PDS}
\dot x=\Pi_{\X}(x,Ax+u).
\end{equation}
Because \(\X\) is closed and convex, and \(Ax + u\) is Lipschitz in \(x\), for a given initial condition,~\eqref{eq:PDS} admits a unique and absolutely continuous trajectory satisfying~\eqref{eq:PDS} a.e. in time~\cite{AN-DZ:96}.

\begin{lemma}[Explicit form of the projection on \(\X\)]
\label{lem:projection_formula}
For every \(x\in\X\), \(v\in\R^n\), and each \(i \in [n]\),
\begin{equation}\label{eq:project_explicit}
    \Pi_{\X}(x,v)_i=
    \begin{cases}
    (v_i)_+, & d_ix_i=0,\\
    v_i, & 0<d_i x_i<1,\\
    -(v_i)_-, & d_i x_i=1.
    \end{cases}
\end{equation}
\end{lemma}
\begin{proof}
Notice that for any \(x\in\X\), the tangent cone is \(T_{\X}(x)=\prod_{i=1}^n K_i(x_i)\), where, for each \(i \in [n]\),
\[
K_i(x_i)=
\begin{cases}
[0,\infty), & d_ix_i=0,\\
\R, & 0<d_i x_i<1,\\
(-\infty,0], & d_i x_i=1.
\end{cases}
\]
The tangent cone is a Cartesian product of intervals, so the projection as defined in~\eqref{eq:projection_op} is coordinatewise. The projection of \(v_i\) onto \(\R\), \([0,\infty)\), and \((-\infty,0]\) gives exactly~\eqref{eq:project_explicit}.
\end{proof}

The system~\eqref{eq:PDS} thus is but a linear system in the interior of \(\X\), and slides along the boundary of \(\X\) without ever leaving \(\X\) when trajectories hit the boundary. 
We next prove that the~\eqref{eq:PDS} is indeed the pointwise fast limit of~\eqref{eq:t-LTN} as \(\tau \to 0^+\).

\begin{proposition}[Fast limit of $\tau$-LTN is a PDS]
\label{prop:fast_limit}
For every \(x\in\X\),
\[
\lim_{\tau\to 0^+} f_\tau(x)=\Pi_{\X}(x,Ax+u).
\]
Hence~\eqref{eq:t-LTN} converges pointwise to~\eqref{eq:PDS} as \(\tau\to0^+\).
\end{proposition}
\begin{proof}
For any \(x\in\X\), consider an arbitrary component $i$ and let \(y_i=d_i x_i\in[0,1]\) and \(g_i=(Ax+u)_i\). 
Then \(f_{\tau,i}(x) \coloneqq (f_{\tau}(x))_i=(1/\tau)\left(-y_i+[y_i+\tau g_i]_0^1\right)
\).
We proceed case by case.
If \(0<y_i<1\), then \(f_{\tau,i}(x)\to g_i\); 
if \(y_i=0\), then \(f_{\tau,i}(x)\to(g_i)_+\); 
if \(y_i=1\), then \(f_{\tau,i}(x)\to -(g_i)_-\), each for \(\tau \to 0^+\).
By Lemma~\ref{lem:projection_formula}, this is exactly \(\Pi_{\X}(x,Ax+u)_i\).
\end{proof}

Next, we show that~\eqref{eq:PDS} and~\eqref{eq:t-LTN} share equilibria.
\begin{proposition}\longthmtitle{PDS and LTN have the same equilibria}
\label{prop:PDS_LTN_same_eq}
A point \(x^\star\in\X\) is an equilibrium of~\eqref{eq:PDS} if and only if it is an equilibrium of~\eqref{eq:LTN}.
\end{proposition}
\begin{proof}
Notice that the equilibrium conditions on the explicit PDS dynamics (substituting \(v = Ax + u\)) in Lemma~\ref{lem:projection_formula} is exactly~\eqref{eq:eq_cond_LTN}, the equilibrium condition for~\eqref{eq:t-LTN}, of which~\eqref{eq:LTN} is a member.
\end{proof}

Propositions~\ref{prop:fast_limit} and~\ref{prop:PDS_LTN_same_eq} explain the relevance of 
of~\eqref{eq:PDS} to~\eqref{eq:t-LTN}.
Recall that,  under the LDS condition,  LTN (and therefore PDS) has a unique equilibrium.
We next prove the global exponential stability (GES) of unique equilibrium~\(x^\star\) of~\eqref{eq:PDS} under the LDS condition.

\begin{theorem}\longthmtitle{GES of the fast PDS limit under LDS}\label{thm:PDS_GES}
Assume \(A\in\mc LDS\) with diagonal certificate \(\Lambda \succ 0\). 
Then~\eqref{eq:PDS} has a unique equilibrium \(x^\star\in\X\), and for each initial condition \(x(0) \in \X\), for all \(t\geq0\)
\[
\lVert x(t)-x^\star\rVert_{\Lambda}\le \mathrm{e}^{-\mu t} \lVert x(0)-x^\star\rVert_{\Lambda},
\]
for some \(\mu>0\). 
Thus \(x^\star\) is globally exponentially stable.
\end{theorem}
\begin{proof}
By Propositions~\ref{prop:common_equilibria} and~\ref{prop:PDS_LTN_same_eq},~\eqref{eq:PDS} has a unique equilibrium \(x^\star\in\X\).
Since \(A\in\mc LDS\), there exists a diagonal \(\Lambda=\diag(\lambda_1,\dots,\lambda_n)\succ0\) such that
\(A^\top\Lambda+\Lambda A\prec0\).
Thus for \(\mu = -\frac12\,\lambda_{\max}\left(\Lambda^{-1/2}(A^\top\Lambda+\Lambda A)\Lambda^{-1/2}\right)\) it holds that
\begin{equation}\label{eq:LDS_mu_fast}
z^\top(A^\top\Lambda+\Lambda A)z
\le -2\mu\, z^\top\Lambda z ,
\quad \forall z\in\R^n.
\end{equation}

Let \(x:[0,\infty) \rightarrow \X\) be any absolutely continuous solution of~\eqref{eq:PDS} and define
\begin{equation}\label{eq:fast_Lyap}
    V_\Lambda(x)=\frac12 \lVert x-x^\star\rVert_\Lambda^2
\end{equation}
Set \(z=x-x^\star\) and \(v=Ax+u\).
Then, for a.e.\ \(t\), \(\dot V_\Lambda(t)=z^\top \Lambda\dot x(t)\).
Lemma~\ref{lem:projection_formula} gives the exact expressions for \(\dot x_i\).
In each case, \(z_i\dot x_i\le z_i v_i\):
in the interior of \(\X\), equality holds; 
if \(d_ix_i=0\), then \(z_i=-x_i^\star\le0\) and \((v_i)_+\ge v_i\), so \(z_i(v_i)_+\le z_i v_i\);
if \(d_i x_i=1\), then \(z_i=1/d_i-x_i^\star\ge0\) and \(-(v_i)_-\le v_i\). 
Thus \(z_i\dot x_i\le z_i v_i\). Weighting each inequality by \(\lambda_i>0\) and adding up, \(z^\top \Lambda \dot x \leq z^\top \Lambda v\). 
Therefore
\[
\dot V_\Lambda
\le
z^\top\Lambda(Ax+u)
=
z^\top\Lambda A z+z^\top\Lambda(Ax^\star+u).
\]
Now \(x^\star\) the equilibrium, so by Proposition~\ref{prop:PDS_LTN_same_eq}, \(x^\star\) satisfies~\eqref{eq:eq_cond_LTN}.
Since \(x\in\X\), it follows in every coordinate that \(z_i(Ax^\star+u)_i \le 0\) and subsequently \(z^\top\Lambda(Ax^\star+u) \le 0\).
Using~\eqref{eq:LDS_mu_fast},
\[\dot V_\Lambda\le z^\top\Lambda A z = \frac12 z^\top(A^\top\Lambda+\Lambda A)z \le -\mu \lVert z\rVert_\Lambda^2=-2\mu V.\]
Gr\"onwall's inequality~\cite{HK:02} yields \(V(t)\le \mathrm{e}^{-2\mu t}V_\Lambda(0)\), i.e.,
\[\lVert x(t)-x^\star\rVert _\Lambda\le \mathrm{e}^{-\mu t}\lVert x(0)-x^\star\rVert_\Lambda,\]
which proves global exponential stability.
\end{proof}

\medskip

Theorem~\ref{thm:PDS_GES} establishes that the unique equilibrium of the PDS is globally exponentially stable. 
An alternative route to show this result would follow from using step~\eqref{eq:LDS_mu_fast} to claim globally \(\mu\)-strong monotonicity of \(Ax + u\) in the \(\Lambda\) metric, and then invoking~\cite[Theorem 4.3]{DZ-AN:95}. However, our proof displays the relevance of the Lyapunov function~\eqref{eq:fast_Lyap}.
The fast limit of the~\eqref{eq:t-LTN} captures the dissipation of the linear region enforced by the LDS condition, which is one motif that leads to hypothesized stability of the LTN as a whole under LDS.
While the metric~\eqref{eq:fast_Lyap} used as the Lyapunov function in Theorem~\ref{thm:PDS_GES} captures the stability of the PDS limit (and therefore, intuitively, of small \(\tau\) members of the \(\tau\)-LTN family), it fails for arbitrary \(\tau\).
The other timescale limit further reveals an important dynamical consequence of the LDS condition, which we explore next.

\section{Slow limit: hard-selector inclusion}
Here we study the global stability properties of the slow limit of the $\tau$-LTN family. 
Throughout the section, we rescale the time to the slow timescale \(s = t/\tau\). For ease of notation, we denote \((\cdot)' = \tfrac{d(\cdot)}{ds}\).
Let the time-rescaled \(\tau\)-LTN family be given by \(x' = F_\tau(x)\), where
\begin{equation}\label{eq:t-LTN_slow}
     F_\tau(x) = -Dx+[Dx+\tau (Ax + u)]_0^1.
\end{equation}
We introduce the hard-selector system (HSS) inclusion
\begin{equation}\label{eq:HSS}\tag{HSS}
 x' \in -Dx+\mc H(Ax + u) =: \mc F(x),
\end{equation}
where \(\mc H:\R^n \rightarrow \mathfrak B([0,1]^n)\) is defined in~\eqref{eq:hard_sel_cvx}.
Under the mild assumption of \(A\) being invertible (automatic under LDS), we show that 
\(\mc F\)
is the Filippov regularization of the \(\tau \to +\infty\) pointwise limit of the time-rescaled \(\tau\)-LTN family vector field \(F_\tau(x)\).
We establish these facts below.
\begin{proposition}\longthmtitle{HSS is the Filippov regularization of the rescaled slow limit}
\label{prop:t-LTN_to_slow_limit}
Let \(F_\infty: \X \rightarrow \R^n\) be defined componentwise as
\begin{equation}\label{eq:slow_pointwise}
(F_\infty(x))_i = \begin{cases}
    - d_ix_i, & A_i^\top x + u_i < 0\\
    0, & A_i^\top x + u_i = 0\\
    1 - d_ix_i, & A_i^\top x + u_i > 0
\end{cases}
\end{equation}
for each \(i \in [n]\).
The following hold:
    \begin{enumerate}[(i)]
        \item For each \(x \in \X\),
        \(\lim_{\tau\to+\infty} F_\tau(x) = F_\infty(x)\);
        \item Assume \(\det A\neq 0\).
        Then, \(\operatorname{F}[F_\infty]\), the Filippov regularization~\eqref{eq:filippov_reg} of \(F_\infty\) is \(\mc F\) as defined in~\eqref{eq:HSS}.
    \end{enumerate}
\end{proposition}

\begin{proof}
    Let \(g(x) = Ax + u\) and let \(g_i(x) = A_i^\top x + u_i\) be its \(i\)th component. 
    For (i), notice that \((F_\tau(x))_i = -d_ix_i + [d_ix_i + \tau g_i(x)]_0^1\).
    If \(g_i(x)>0\), \([d_ix_i + \tau g_i(x)]_0^1 = 1\) for sufficiently large \(\tau\).
    Similarly, if \(g_i(x)<0\), \([d_ix_i + \tau g_i(x)]_0^1 = 0\) for sufficiently large \(\tau\).
    If \(g_i(x) = 0\), and since \(d_ix_i \in [0,1]\), \([d_ix_i + \tau g_i(x)]_0^1 = d_ix_i\).
    Thus, in every case, \(\lim_{\tau\to+\infty} F_\tau(x) = F_\infty(x)\).

    For (ii), we first show that, for each \(x\in\X\), \(\operatorname{F}[F_\infty](x) \subseteq -Dx + \mc H(g(x))\).
    Suppose \(g_i(x) > 0\). 
    By the continuity of \(g_i\) in \(\R^n\), for all \(y\) sufficiently close to \(x\), \((F_\infty(y))_i = 1 - d_iy_i\) and thus \((\operatorname{F}[F_\infty](x))_i = 1 - d_ix_i\).
    Similarly, for \(g_i(x)< 0\), \((\operatorname{F}[F_\infty](x))_i = - d_ix_i\).
    When \(g_i(x) = 0\), we may have points \(y\) such that \(g_i(y)>0\) or \(g_i(y)<0\), so that \((F_\infty(y))_i = 1 - d_iy_i\), \(- d_iy_i\), or \(0\). 
    These are contained in interval \(-d_ix_i + [0,\,1]\), which is their convex closure. 
    Thus, \(\operatorname{F}[F_\infty](x) \subseteq -Dx + \mc H(g(x))\).
    
    We now consider the reverse inclusion.
    Let \(I:=\{i\in[n]:g_i(x)=0\}\). 
    Since \(A\) is invertible, the submatrix \(A_I\) has full row rank.
    Hence for any corner of \(-Dx+\mc H(g(x))\), there exists \(h\in\R^n\) such that \(A_i^\top h>0\) whenever \(\eta_i=1-d_i x_i\) and \(A_i^\top h<0\) whenever \(\eta_i=-d_i x_i\), for all \(i\in I\). 
    Then for \(y_t:=x+th\), one has \(g_i(y_t)=tA_i^\top h\) for \(i\in I\), so its sign agrees with the one corresponding to \(\eta_i\); for \(g_i(x)\neq0\), continuity preserves the sign for all sufficiently small \(t>0\). 
    Thus \(F_\infty(y_t)=\eta\) for all sufficiently small \(t>0\), and in fact on a nonempty open set arbitrarily close to \(x\). 
    This set has positive measure since \(h\) can be perturbed to generate an open set, while still generating the same strict sign in \(A_i^\top \tilde h \lessgtr 0\) for  \(\tilde h\) belonging to that open set, since \(y\mapsto Ay\) is continuous.
    Therefore, removing a null set cannot eliminate all such points, and \(\eta\in \operatorname{F}[F_\infty](x)\). 
    Since \(\operatorname{F}[F_\infty](x)\) is closed and convex (being the intersection of convex closures), it contains the convex hull of all corners of \(-Dx+\mc H(g(x))\), namely \(-Dx+\mc H(g(x))\). 
    Thus \(-Dx+\mc H(g(x))\subseteq \operatorname{F}[F_\infty](x)\).
    Combining with the first inclusion gives \(\operatorname{F}[F_\infty](x)= -Dx+\mc H(g(x))\).
\end{proof}
\begin{remark}
    In the proof above, why does the reverse inclusion require the invertibility of \(A\)? 
    The Filippov regularization of \(F_\infty\) may be strictly smaller than~\eqref{eq:HSS} if \(A\) loses rank and each component cannot independently generate all sign patterns.
    A simple example is furnished.
    Let \(A = [1,\,0;\, 1,\, 0]\) with \(u=0\). Then, at \(x=0\), \(\operatorname{F}[F_\infty](0) = \overline{\operatorname{co}}(\{(0,\,0),\, (1,\,1)\})\) collapses to a line while \(\mc H(0) = [0,1]^2\); the former cannot contain the latter.
\end{remark}

Thus, importantly, a Filippov solution of the discontinuous slow limit is an absolutely continuous function \(x:[0,\infty) \rightarrow \R^n\) satisfying the inclusion~\eqref{eq:HSS}.
Similar to the fast limit, we first prove that the convexified selector system~\eqref{eq:HSS} indeed contains the pointwise slow limit of~\eqref{eq:t-LTN} reparametrized to slow time.
We comment on the existence and forward invariance of the slow-limit system solutions on \(\X\) below.
\begin{proposition}\longthmtitle{Existence of slow-limit Filippov solutions and strong forward invariance}
\label{prop:slow_existence_invariance}
Let \(\det A \neq 0\). 
Then, for each \(x(0) \in \X\), the following hold:
\begin{enumerate}[(i)]
    \item \(x'(s) = F_\infty(x(s))\) possesses Filippov solutions given by absolutely continuous trajectories \(s\mapsto x(s)\in \R^n\) satisfying~\eqref{eq:HSS} for a.e. \(s\geq 0\), and
    \item \(\X\) is strongly forward invariant for Filippov solutions of \(x'(s) = F_\infty(x(s))\).
\end{enumerate}
\end{proposition}
\begin{proof} 
\(F_\infty\) takes bounded values everywhere on \(\X\). 
In particular, it is locally essentially bounded.
Since \(A\) is invertible, each set \(\{x \in \X :A_i^\top x + u_i= 0\}\) is either empty or is an affine hyperplane, and is of measure zero. 
Consequently, \(\bigcup_{i\in[n]}\{x \in \X :A_i^\top x + u_i= 0\}\) has measure zero.
Thus, \(F_\infty\) is continuous a.e. on \(\X\) and is measurable on \(\X\).
By~\cite[Proposition 3]{JC:08}, \(x'(s) = F_\infty(x(s))\) has Filippov solutions.
By Proposition~\ref{prop:t-LTN_to_slow_limit} (i), under \(\det A \neq 0\), \(\operatorname{F}[F_\infty] = \mc F\), and by definition, absolutely continuous solutions satisfying \(x'(s) \in \mc F(x(s))\) a.e. \(s\) are exactly Filippov solutions of \(x'(s) = F_\infty(x(s))\).

For (ii), if \(x_i=0\), then \(x'_i=\sigma_i\in[0,1]\) for a.e.\ time. If \(d_i x_i=1\), then
\(\dot x_i=\sigma_i-1\in[-1,0]\)
for a.e.\ time. Hence no coordinate can leave the interval \([0,1/d_i]\), and all possible solutions from an initial condition are forward invariant, giving us the claim.
\end{proof}

We note at this juncture that~\eqref{eq:HSS} has not been established to admit unique trajectories; it probably does not in general for arbitrary \(D \in \mathbb D_+^n\), \(W\in \R^{n\times n}\), especially when neighboring regions are repulsive at the shared boundary.

\begin{proposition}\longthmtitle{HSS and LTN have the same equilibria}
\label{prop:HSS_equilibria}
The following are true:
\begin{enumerate}[(i)]
    \item~\eqref{eq:HSS} and~\eqref{eq:LTN} have the same equilibria, and thus so does the~\eqref{eq:t-LTN} family and its degenerate endpoints.
    \item Assume \(A \in \mc LDS\).~\eqref{eq:HSS} admits a unique equilibrium.
\end{enumerate}
\end{proposition}

\begin{proof}
Allow \(g(x) = Ax + u\), and let \(g_i\) be its \(i\)th component.
A point \(x\in\X\) is an equilibrium of~\eqref{eq:HSS} iff there exists \(\sigma\in\mc H(g(x))\) such that \(\sigma=Dx\).
Coordinatewise, if \(g_i(x)>0\), then \(\sigma_i=1\), so \(d_i x_i=1\). 
If \(g_i(x)<0\), then \(\sigma_i=0\), so \(d_ix_i=0\).
If \(g_i(x)=0\), then \(\sigma_i\in[0,1]\), so \(d_i x_i\in[0,1]\). 
This is exactly the equilibrium condition as for~\eqref{eq:LTN} as given in~\eqref{eq:eq_cond_LTN}.
For (ii), assume \(A\in \mc LDS\). 
Then, by Proposition~\ref{prop:common_equilibria},~\eqref{eq:LTN} admits a unique equilibrium, and by~(i), we are done.
\end{proof}

Assuming  \(A \in \mc LDS\), let \(x^\star\) be the unique equilibrium of~\eqref{eq:HSS}.
To discuss its global stability, define the  nonsmooth candidate Lyapunov function
\begin{equation}\label{eq:Vslow_def}
V_\infty(x):= \max_{\zeta \in \{0, 1\}^n}
(Ax + u)^\top \Lambda (\zeta-Dx)
\end{equation}
where \(\Lambda=\diag(\lambda_1,\dots,\lambda_n)\succ0\) is the diagonal matrix from the LDS condition on \(A\).
We will now discuss the validity of \(V_\infty\) as a Lyapunov function.

\begin{proposition}\longthmtitle{Slow-limit Lyapunov function}
\label{prop:Vslow_properties}
The function \(V_\infty\) defined in~\eqref{eq:Vslow_def} satisfies the following properties:
\begin{enumerate}[i)]
\item For every \(x\in\X\),
\begin{align*}
    V_\infty(x)
    =
    \sum_{i=1}^n
    \lambda_i\Big((A_i^\top x +& u_i)_+(1-d_i x_i)\\&+(A_i^\top x + u_i)_-\,d_i x_i\Big).
\end{align*}
Equivalently,
\[
V_\infty(x)
=
\sum_{i=1}^n
\lambda_i\Big((A_i^\top x + u_i)_+-d_i (A_i^\top x + u_i)x_i\Big).
\]
\item \(V_\infty\) is globally Lipschitz on \(\X\) and satisfies \(V_\infty(x)\ge0\) for all \(x\in\X\).
\item \(V_\infty(x)=0 \iff x\) is an equilibrium of~\eqref{eq:HSS}.
In particular, under \(A \in \mc LDS\), \(V_\infty\) is positive definite with respect to the unique equilibrium \(x^*\).
\end{enumerate}
\end{proposition}

\begin{proof}
Define \(g(x) = Ax + u\) and its \(i\)th component by \(g_i(x) = A_i^\top x + u_i\).
Thus, for each \(i\), one maximizes \(\lambda_i g_i(x)(\zeta_i-d_i x_i)\) over \(\zeta_i\in\{0,1\}\). 
If \(g_i(x)>0\), the maximizer is \(\zeta_i=1\). 
If \(g_i(x)<0\), the maximizer is \(\zeta_i=0\). 
If \(g_i(x)=0\), that component does not contribute to \(V_\infty\) and can be ignored. 
This yields
\(V_\infty(x)
=
\sum_{i=1}^n
\lambda_i\left((g_i(x))_+(1-d_i x_i)+(g_i(x))_-\,d_i x_i\right)\).
The alternative form follows from \(g_i = (g_i)_+ - (g_i)_-\).

By the explicit formula, \(V_\infty\) is a finite sum of terms of the form \(x\mapsto \lambda_i(g_i(x))_+\) and \(x\mapsto -\lambda_i d_i g_i(x)x_i\). Since \(g_i\) is affine, \(r\mapsto r_+\) is Lipschitz, and \(\X\) is compact, \(V_\infty\) is globally Lipschitz on \(\X\). Nonnegativity is immediate from the explicit formula, as \(d_ix_i \in [0, 1]\).

Finally, \(V_\infty(x)=0\) if and only if, for every \(i\), \((g_i(x))_+(1-d_i x_i)=0\) and \((g_i(x))_-\,d_i x_i=0\) since both are nonnegative.
Equivalently, if \(g_i(x)>0\), then \(d_i x_i=1\); if \(g_i(x)<0\), then \(x_i=0\); and if \(0<d_i x_i<1\), then \(g_i(x)=0\). 
By Proposition~\ref{prop:HSS_equilibria}, this is exactly the equilibrium condition. 
Under \(A \in \mc LDS\), since by Proposition~\ref{prop:common_equilibria} the LTN equilibrium is unique, the last claim follows.
\end{proof}

We now show that the Lyapunov function~\eqref{eq:Vslow_def} is nonincreasing for almost all slow time, along trajectories of~\eqref{eq:HSS}.
\begin{theorem}\longthmtitle{Strong GAS of the slow-selector limit under LDS}
\label{thm:slow_GAS}
Assume \(A\in\mc LDS\). Then the unique equilibrium \(x^\star\) of~\eqref{eq:HSS} is strongly globally asymptotically stable. More precisely, for \(d_{\min}:=\min_{i\in[n]} d_i>0\), for all \(0\le s_1\le s_2\),
\begin{equation}\label{eq:slowLyapunov-decay}
V_\infty(x(s_2))\le \mathrm{e}^{-d_{\min}(s_2-s_1)}V_\infty(x(s_1)).
\end{equation}
\end{theorem}
\vspace{11pt}
\begin{proof}
Let \(x:[0,\infty) \rightarrow \X\) be any absolutely continuous solution of~\eqref{eq:HSS} satisfying it a.e. \(s\). 
Since \(V_\infty\) is Lipschitz (Proposition~\ref{prop:Vslow_properties}), \(s\mapsto V_\infty(x(s))\) is a.e. differentiable \cite[Proposition 10]{JC:08}.

Let \(g_i(s):= A_i^\top x(s)+u_i\).
To calculate the Lyapunov derivative, we first compute \(g_i'(s)\).
Since the solution trajectory \(x\) satisfies~\eqref{eq:HSS} for a.e.\ \(s\), for each \(i\) and for a.e.\ \(s\), \(g_i(s)>0 \Longrightarrow x_i'(s)=1-d_i x_i(s)\) and \(g_i(s)<0 \Longrightarrow x_i'(s)=-d_i x_i(s)\).
Also, since \(g_i\) is absolutely continuous, one has \(g_i'(s)=0\) for a.e.\ \(s\) such that \(g_i(s) = 0\). Differentiating \(V_\infty(x(s))\) for a.e.\ \(s\) gives
\(V'_\infty(x(s)) = \sum_{i=1}^n \lambda_i p_i(s)\); 
\[p_i(s) = \left(\mathbf 1_{\{g_i(s)>0\}}-d_i x_i(s)\right)g_i'(s)-d_i g_i(s) x'_i(s).\]
We observe that \((\mathbf 1_{\{g_i(s)>0\}}-d_i x_i(s)) g_i'(s) = x'_i(s) g_i'(s)\) for a.e. s, as, if \(g_i>0\) \(x'_i=1-d_i x_i\) and if \(g_i<0\) \(x'_i=-d_i x_i\) and when \(g_i=0\), both sides vanish a.e.
Using \(g_i'=A_i^\top x'\),
\begin{align*}
    V'_\infty(x(s))
    &=
    \sum_{i=1}^n \lambda_i\bigl(x'_i(s) g_i'(s)-d_i g_i(s) x'_i(s)\bigr)\\
    &=
    x'(s)^\top\Lambda A\, x'(s)
    -
    \sum_{i=1}^n d_i\lambda_i g_i(s)x'_i(s)
\end{align*}
for a.e.\ \(s\).
We now tackle the second term. If \(g_i(s)>0\), then \(g_i(s)x'_i(s)=g_i(s)(1-d_i x_i(s))=(g_i(s))_+(1-d_i x_i(s))\), which is nonnegative. 
If \(g_i(s)<0\), then \(g_i(s) x'_i(s)=g_i(s)(-d_i x_i(s))=(g_i(s))_-d_i x_i(s)\) which is also nonnegative. 
If \(g_i(s)=0\), both sides vanish. 
Therefore, comparing with the explicit form in Proposition~\ref{prop:Vslow_properties} (i), \(\sum_{i=1}^n \lambda_i g_i(s) x'_i(s)=V_\infty(x(s))\) for a.e. \(s\).
Since each summand is nonnegative,
\[
\sum_{i=1}^n d_i\lambda_i g_i(s) x'_i(s)
\ge
d_{\min}\sum_{i=1}^n \lambda_i g_i(s) x'_i(s)
=
d_{\min}V_\infty(x(s)).
\]
Hence \(V'_\infty(x(s)) \le x'(s)^\top\Lambda A\, x'(s)-d_{\min}V_\infty(x(s))\) for a.e. \(s\).
Under the LDS hypothesis, by~\eqref{eq:LDS_mu_fast} and setting \(z = x'(s)\),
\[
x'(s)^\top\Lambda A\, x'(s)
\le -\mu \lVert x'(s)\rVert_\Lambda^2,
\]
\[ \implies
V'_\infty(x(s))
\le
-\mu \lVert x'(s)\rVert_\Lambda^2-d_{\min}V_\infty(x(s))
\quad\text{for a.e.\ }s.
\]

Dropping the nonnegative term \(\mu\lVert x'\rVert_\Lambda^2\) and multiplying by \(\mathrm{e}^{d_{\min} s}\) gives
\(\left(\mathrm{e}^{d_{\min} s}V_\infty(x(s))\right)'\le 0\) for a.e.\ \(s\).
\[\implies V_\infty(x(s_2))\le \mathrm{e}^{-d_{\min}(s_2-s_1)}V_\infty(x(s_1)).\]
This gives us \(V_\infty(x(s))\to0\) as \(s\to\infty\).
Since \(V_\infty\) is continuous on the compact set \(\X\) and, by Proposition~\ref{prop:Vslow_properties} (iii), is positive definite with respect to the unique equilibrium \(x^*\) under the LDS condition, for each \(\varepsilon>0\) it holds that
\[m_\varepsilon := \min\{V_\infty(x):x\in\X,\ \lVert x-x^\star\rVert\ge\varepsilon\} > 0.\]

Thus, as \(V_\infty\) vanishes, so must \(\lVert x-x^\star\rVert\), and \(x(s)\to x^\star\).
It remains to prove Lyapunov stability. Fix \(\varepsilon>0\) and thus \(m_\varepsilon\).
By continuity of \(V_\infty\) at \(x^\star\), there exists \(\delta>0\) such that \(\lVert x-x^\star\rVert<\delta\) so that
\(V_\infty(x)<m_\varepsilon\).
If \(\lVert x(0)-x^\star\rVert<\delta\), then \(V_\infty(x(0))<m_\varepsilon\), and~\eqref{eq:slowLyapunov-decay} implies that, for all \(s\ge0\),
\[V_\infty(x(s))\le V_\infty(x(0))<m_\varepsilon\]
Hence \(\lVert x(s)-x^\star\rVert<\varepsilon\) for all \(s\ge0\). 
Therefore \(x^\star\) is Lyapunov stable. 
Combined with convergence, this proves strong global asymptotic stability.
\end{proof}

We have established that the slow endpoint also possesses a globally asymptotically stable equilibrium, and have identified the Lyapunov function that establishes this.

\section{Discussion}

Having established the global stability properties of the two endpoints of the $\tau$-LTN family, we note that this does not yet prove global asymptotic stability for every intermediate member, and therefore do not resolve Conjecture~\ref{conj:LDS-GAS}. 
Theorems~\ref{thm:PDS_GES} and~\ref{thm:slow_GAS} do, however, isolate two stability mechanisms that any technical argument for the whole family must reconcile: contraction of the underlying linear drift in a diagonal metric, and coherent switching induced by saturation. 
In this sense, the $\tau$-LTN family is not merely a homotopy between models, but it distills the LTN dynamics into two isolated and analytically tractable motifs whose interplay appears central to LTN global stability under LDS.
The constructed family allows us to view LTN dynamics as a combination of structured linear dissipation and a combinatorial stable switching under the LDS hypothesis.

This viewpoint also suggests a useful interpretation of $\tau$ as a boundary-sensitivity parameter. 
When $\tau$ is small, trajectories behave as if they are relatively insensitive or myopic to the boundary of $\X$: they move according to the linear drift until projection becomes necessary.
When $\tau$ is large, the dynamics respond to the impending saturation much earlier, effectively turning according to the sign pattern of $Ax+u$ before the full affine drift is realized. 
The canonical LTN at $\tau=1$ sits between these two extremes.%
We perform simulations to show the result of varying \(\tau\).
Figure~\ref{fig:trajectories} displays a Lyapunov descent study where we use either endpoint certificate to attempt certifying stability numerically.
The simulations demonstrate the suitability of the Lyapunov functions for the endpoints' trajectories.

We also performed numerical experiments outside the LDS regime, considering the 2-dimensional case of limit cycle emergence, with parameters chosen following~\cite{EN-RP-JC:22}. 
The results are presented in Figure~\ref{fig:limit_cycle}.
Surprisingly, the \(\tau\)-LTN family was empirically observed to retain the oscillatory behavior throughout the family and at the endpoints as well.
This gives us further evidence that the \(\tau\)-LTN family exhibits shared dynamical behavior across all elements of the family.

\begin{figure}
    \centering
    \includegraphics[width=1\linewidth]{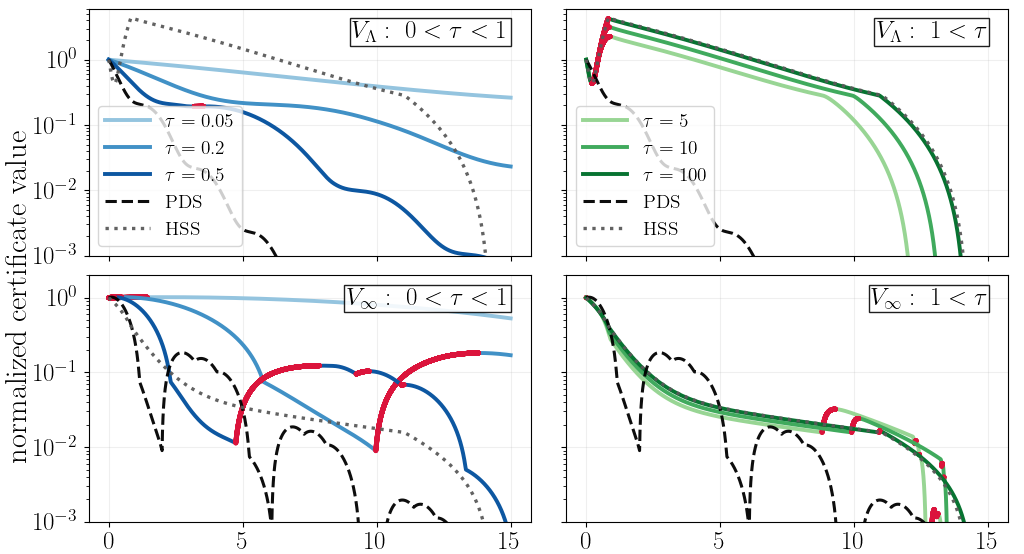}
    \caption{
    Lyapunov certificate descent study under Lyapunov diagonal stability using the two discovered Lyapunov functions \(V_\Lambda\) and \(V_\infty\) with trajectories of different members of the \(\tau\)-LTN family and the endpoints.
    Here, \(W=[0.0,\, -1.6;\, 1.6,\, 0.0]\), \(D = \diag(1.0,\,0.1)\), and \(u = [0.5,\, -0.2]^\top\)
    Empirically shows the expected fact that members close to each limit should be compatible with the corresponding limit's certificate.
    Red demarcations indicate Lyapunov descent violation, diagonal plots show decreasingly common violations as \(\tau\) gets closer to the respective limit.
    The anti-diagonal plots show increasingly severe violations, as trajectories are more aligned to the respective other endpoint's dynamics.
    ``Normalized" here represents the Lyapunov evolution along trajectories as a fraction of its initial value.
    The HSS descent stops at a set tolerance due to numerical issues.
    }
    \label{fig:trajectories}
\end{figure}

\begin{figure*}
    \centering
    \includegraphics[width=0.9\textwidth]{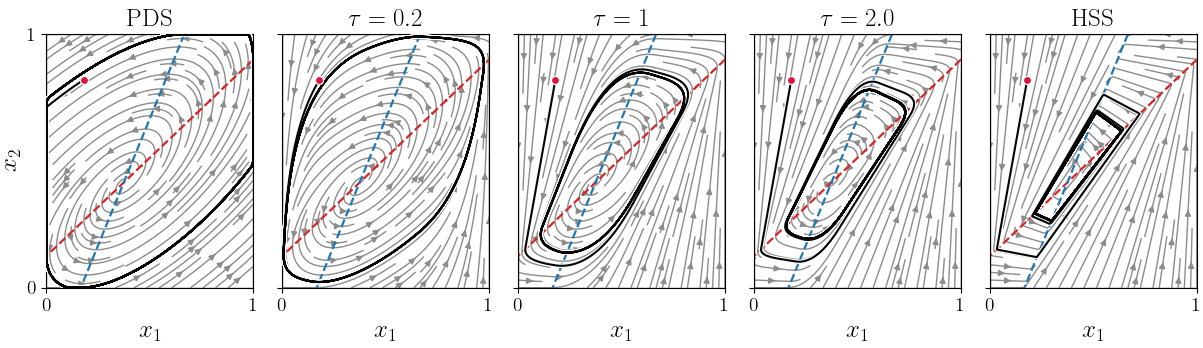}
    \caption{Empirical simulation evidence that oscillatory behavior present in a non-Lyapunov diagonally stable example is reflected in both endpoint regimes of the $\tau$-linear-threshold network ($\tau$-LTN) family, and in the endpoints themselves. 
    Middle three panels: members of the \(\tau\)-LTN family for 
    \(\tau \in \{0.2,\, 1,\,2.0\}\); 
    extreme panels: the degenerate projected dynamical system (PDS) and hard-selector system (HSS) limits.
    This qualitatively supports the claim that the family preserves a common dynamical skeleton across timescales and in the degenerate limit.
    Red dashed line: locus of \((Ax + u)_1 = 0\); the blue one of \((Ax + u)_2 = 0\).
    Here, \(W = [4.1,\, -4.0;\, 3.0,\, -0.5]\), \(D = \diag(1,\, 1)\), and \(u = [0.5,\, -0.5]\).
    Parameters are chosen following~\cite[Theorem 3.1]{EN-RP-JC:22}.}
    \label{fig:limit_cycle}
\end{figure*}

\section{Conclusions}

We introduced the \(\tau\)-LTN family as a one-parameter deformation of the canonical linear-threshold network that preserves both Lyapunov diagonal stability (LDS) and the equilibrium set. The \(\tau\)-LTN family reveals two analytically tractable limiting regimes. In the fast limit, the dynamics converge to a projected dynamical system for which LDS ensures global exponential stability. In the slow limit, the dynamics converge to a hard-selector differential inclusion for which LDS guarantees global asymptotic stability. These limits expose the fundamental dissipative and switching mechanisms governing LTN behavior.

Future work will focus on establishing global asymptotic stability for the full \(\tau\)-LTN family. Viable approaches are either a Lyapunov-based proof requiring an interpolation of the endpoint Lyapunov functions or continuation arguments. The \(\tau\)-family also exhibits striking regularity-preservation of LDS and persistence of qualitative behaviors, which merits further analytical investigation. These directions may extend to broader activation functions and contribute to a more complete theory of biologically plausible nonlinear networks.

\appendices
\counterwithin{theorem}{section}
\renewcommand{\thetheorem}{A\thesection.\arabic{theorem}}

\section{Linear-threshold network well-posedness}
\begin{lemma}[LTN forward invariance]
\label{lem:forward_invariance}
The set \(\X\) is forward invariant under the dynamics~\eqref{eq:LTN}, that is, for each \(x(0) \in \X\), the solution of~\eqref{eq:LTN} remains in \(\X\) for all time.
\end{lemma}
\begin{proof}
Since the vector field of~\eqref{eq:LTN} is locally Lipschitz (\(-Dx\) being linear, and \([Wx + u]_0^1\) being the composition of a linear function with a \(1\)-Lipschitz nonlinearity), existence and uniqueness of solution \(x:[0, \infty) \rightarrow \R^n\) starting from any \(x(0) \in \X\) is guaranteed, cf.~\cite{HK:02}.
For the \(i\)th coordinate, \( \dot x_i= -d_ix_i+\left([Wx+u]_0^1\right)_i\) with each \(x_i \in [0,1/d_i]\). If \(d_ix_i=0\), then \( \dot x_i= \left([Wx+u]_0^1\right)_i\ge0\). 
If \(d_ix_i=1\), then \( \dot x_i= -1+\left([Wx+u]_0^1\right)_i\le0\). Hence no coordinate \(x_i\) can leave \([0,1/d_i]\), so trajectories starting in \(\X\) remain in \(\X\).
\end{proof}

\smallskip
\bibliographystyle{IEEEtran}%
\bibliography{ref, SB-main}

\begin{thebibliography}{10}
\providecommand{\url}[1]{#1}
\csname url@samestyle\endcsname
\providecommand{\newblock}{\relax}
\providecommand{\bibinfo}[2]{#2}
\providecommand{\BIBentrySTDinterwordspacing}{\spaceskip=0pt\relax}
\providecommand{\BIBentryALTinterwordstretchfactor}{4}
\providecommand{\BIBentryALTinterwordspacing}{\spaceskip=\fontdimen2\font plus
\BIBentryALTinterwordstretchfactor\fontdimen3\font minus \fontdimen4\font\relax}
\providecommand{\BIBforeignlanguage}[2]{{%
\expandafter\ifx\csname l@#1\endcsname\relax
\typeout{** WARNING: IEEEtran.bst: No hyphenation pattern has been}%
\typeout{** loaded for the language `#1'. Using the pattern for}%
\typeout{** the default language instead.}%
\else
\language=\csname l@#1\endcsname
\fi
#2}}
\providecommand{\BIBdecl}{\relax}
\BIBdecl

\bibitem{CMA-GS:83}
M.~A. Cohen and S.~Grossberg, ``Absolute stability of global pattern formation and parallel memory storage by competitive neural networks,'' \emph{IEEE Transactions on Systems, Man, and Cybernetics}, vol. SMC-13, no.~5, p. 815–826, Sep. 1983.

\bibitem{HJJ:82}
J.~J. Hopfield, ``Neural networks and physical systems with emergent collective computational abilities.'' \emph{Proceedings of the National Academy of Sciences}, vol.~79, no.~8, p. 2554–2558, Apr. 1982.

\bibitem{HJJ:84}
------, ``Neurons with graded response have collective computational properties like those of two-state neurons.'' \emph{Proceedings of the National Academy of Sciences}, vol.~81, no.~10, p. 3088–3092, May 1984.

\bibitem{EN-JC:21}
\BIBentryALTinterwordspacing
E.~Nozari and J.~Cort\'es, ``Hierarchical selective recruitment in linear-threshold brain networks—{P}art {I}: Single-layer dynamics and selective inhibition,'' \emph{IEEE Transactions on Automatic Control}, vol.~66, no.~3, p. 949–964, Mar. 2021. [Online]. Available: \url{http://dx.doi.org/10.1109/TAC.2020.3004801}
\BIBentrySTDinterwordspacing

\bibitem{EN-RP-JC:22}
\BIBentryALTinterwordspacing
E.~Nozari, R.~Planas, and J.~Cortés, ``Structural characterization of oscillations in brain networks with rate dynamics,'' \emph{Automatica}, vol. 146, p. 110653, Dec. 2022. [Online]. Available: \url{http://dx.doi.org/10.1016/j.automatica.2022.110653}
\BIBentrySTDinterwordspacing

\bibitem{KM-AD-VI-CC:24}
\BIBentryALTinterwordspacing
K.~Morrison, A.~Degeratu, V.~Itskov, and C.~Curto, ``Diversity of emergent dynamics in competitive threshold-linear networks,'' \emph{SIAM Journal on Applied Dynamical Systems}, vol.~23, no.~1, p. 855–884, Mar. 2024. [Online]. Available: \url{http://dx.doi.org/10.1137/22M1541666}
\BIBentrySTDinterwordspacing

\bibitem{EK-AB:00}
\BIBentryALTinterwordspacing
E.~Kaszkurewicz and A.~Bhaya, \emph{\BIBforeignlanguage{en}{Matrix {Diagonal} {Stability} in {Systems} and {Computation}}}.\hskip 1em plus 0.5em minus 0.4em\relax Boston, MA: Birkhäuser Boston, 2000. [Online]. Available: \url{http://link.springer.com/10.1007/978-1-4612-1346-8}
\BIBentrySTDinterwordspacing

\bibitem{MA-CM-AP:16}
\BIBentryALTinterwordspacing
M.~Arcak, C.~Meissen, and A.~Packard, \emph{\BIBforeignlanguage{en}{Networks of {Dissipative} {Systems}}}, ser. {SpringerBriefs} in {Electrical} and {Computer} {Engineering}.\hskip 1em plus 0.5em minus 0.4em\relax Cham: Springer International Publishing, 2016. [Online]. Available: \url{http://link.springer.com/10.1007/978-3-319-29928-0}
\BIBentrySTDinterwordspacing

\bibitem{FM-TA:95}
M.~Forti and A.~Tesi, ``New conditions for global stability of neural networks with application to linear and quadratic programming problems,'' \emph{IEEE Transactions on Circuits and Systems I: Fundamental Theory and Applications}, vol.~42, no.~7, p. 354–366, Jul. 1995.

\bibitem{BS-RW-DA:26}
\BIBentryALTinterwordspacing
S.~Betteti, W.~Retnaraj, A.~Davydov, J.~Cortés, and F.~Bullo, 2025. [Online]. Available: \url{https://arxiv.org/abs/2512.05252}
\BIBentrySTDinterwordspacing

\bibitem{PD-AN:93}
\BIBentryALTinterwordspacing
P.~Dupuis and A.~Nagurney, ``Dynamical systems and variational inequalities,'' \emph{Annals of Operations Research}, vol.~44, no.~1, p. 7–42, Feb. 1993. [Online]. Available: \url{http://dx.doi.org/10.1007/BF02073589}
\BIBentrySTDinterwordspacing

\bibitem{AN-DZ:96}
A.~Nagurney and D.~Zhang, \emph{\BIBforeignlanguage{en}{Projected dynamical systems and variational inequalities with applications}}, 1996th~ed., ser. International Series in Operations Research \& Management Science.\hskip 1em plus 0.5em minus 0.4em\relax Dordrecht, Netherlands: Springer, Dec. 1995.

\bibitem{JC:08}
J.~Cort\'es, ``Discontinuous dynamical systems,'' \emph{IEEE Control Systems Magazine}, vol.~28, no.~3, pp. 36--73, 2008.

\bibitem{AFF:88}
\BIBentryALTinterwordspacing
A.~F. Filippov, \emph{Differential Equations with Discontinuous Righthand Sides}, F.~M. Arscott, Ed.\hskip 1em plus 0.5em minus 0.4em\relax Springer Netherlands, 1988. [Online]. Available: \url{http://dx.doi.org/10.1007/978-94-015-7793-9}
\BIBentrySTDinterwordspacing

\bibitem{REK:63}
\BIBentryALTinterwordspacing
R.~E. Kalman, ``Lyapunov functions for the problem of lur’e in automatic control,'' \emph{Proceedings of the National Academy of Sciences}, vol.~49, no.~2, p. 201–205, Feb. 1963. [Online]. Available: \url{http://dx.doi.org/10.1073/pnas.49.2.201}
\BIBentrySTDinterwordspacing

\bibitem{brouwer1911abbildung}
L.~E.~J. Brouwer, ``{\"U}ber abbildung von mannigfaltigkeiten,'' \emph{Mathematische annalen}, vol.~71, no.~1, pp. 97--115, 1911.

\bibitem{HK:02}
H.~Khalil, \emph{Nonlinear Systems}, ser. Pearson Education.\hskip 1em plus 0.5em minus 0.4em\relax Prentice Hall, 2002.

\bibitem{DZ-AN:95}
\BIBentryALTinterwordspacing
D.~Zhang and A.~Nagurney, ``On the stability of projected dynamical systems,'' \emph{Journal of Optimization Theory and Applications}, vol.~85, no.~1, p. 97–124, Apr. 1995. [Online]. Available: \url{http://dx.doi.org/10.1007/BF02192301}
\BIBentrySTDinterwordspacing

\end{thebibliography}

\end{document}